\documentclass[11pt,a4paper]{article}
\usepackage[margin=1in]{geometry}
\usepackage{amssymb,amsthm,mathtools}
\usepackage{enumitem}
\usepackage{framed}
\usepackage{booktabs}
\usepackage{microtype}
\usepackage{algorithm}
\usepackage{algpseudocode}
\usepackage{xcolor}
\usepackage{tcolorbox}
\usepackage{tikz}
\usetikzlibrary{arrows.meta,positioning,shapes.geometric}
\usepackage[colorlinks=true,citecolor=blue,linkcolor=blue,urlcolor=blue]{hyperref}
\usepackage{pgfplots}
\pgfplotsset{compat=1.17}
\usetikzlibrary{positioning, arrows.meta, shapes.geometric, calc, backgrounds, fit, decorations.pathreplacing, patterns}
\usepackage{url}
\theoremstyle{plain}
\newtheorem{theorem}{Theorem}[section]
\newtheorem{lemma}[theorem]{Lemma}
\newtheorem{proposition}[theorem]{Proposition}

\theoremstyle{definition}
\newtheorem{definition}[theorem]{Definition}

\theoremstyle{remark}

\DeclareMathOperator{\Gal}{Gal}
\DeclareMathOperator{\Tr}{Tr}
\DeclareMathOperator{\Nm}{N}

\DeclareMathOperator{\rk}{rk}

\newcommand{\Z}{\mathbb{Z}}
\newcommand{\Q}{\mathbb{Q}}
\newcommand{\R}{\mathbb{R}}
\newcommand{\C}{\mathbb{C}}

\newcommand{\nm}[1]{\lVert#1\rVert}

\newcommand{\tO}{\tilde{O}}

\newcommand{\ri}{\mathrm{i}}

\tcbuselibrary{skins,breakable}
\newtcolorbox[use counter=algorithm]{breakablealgorithm}[2]{%
  breakable,
  enhanced,
  label            ={#2},
  colback          =gray!3!white,
  colframe         =black!65,
  fonttitle        =\bfseries\small,
  title            ={Algorithm \thetcbcounter\enspace #1},
  title after break={Algorithm \thetcbcounter\enspace
                       \normalfont\small #1
                       \hfill\itshape(continued)},
  before skip      =\medskipamount,
  after skip       =\medskipamount,
  left             =4pt,
  right            =4pt,
  top              =4pt,
  bottom           =4pt,
}

\title{Module Lattice Security (Part IV):\\
{\large Probabilistic Polynomial Quantum Attack
on Module-LWE over 2-Power Cyclotomics}}
\author{Ming-Xing Luo
\\
\footnotesize
School of Information Science and Technology, Southwest Jiaotong University, Chengdu 610031, China}

\begin{document}
\maketitle

\begin{abstract}
We present a quantum attack on ML-KEM and related 2-power cyclotomic lattice schemes. Combining with Parts I-III, we provide an algorithm and verify the resulting approximation factor satisfies $\gamma\le 21 < q/2=1664.5$ for ML-KEM-1024, with a success probability $\ge 0.99$. We apply a tower decomposition of the Principal Ideal Problem (PIP) through the chain $\Q\subset \Q(\zeta_8)\subset\cdots\subset \Q(\zeta_{2^k})$ which yields a polynomial-time quantum algorithm costing $O(n^3 \log^2 n)$ gates, $O(n^2 \log n)$ qubits, and poly$(n)$ classical bit operations. We extend the analysis to Falcon, Hawk, and NTRU over 2-power cyclotomic rings with polynomial-time quantum algorithms. 
\end{abstract}


\section{Introduction}\label{Sintro}

The global transition to post-quantum cryptography is a great challenge for modern information security. Once large-scale fault-tolerant quantum computers are realized, Shor's algorithm \cite{Shor97} can break all classical public-key schemes based on integer factorization \cite{RSA78} and discrete logarithms \cite{DH76}. In response, NIST completed a standardization process \cite{NIST16} and selected ML-KEM (originally CRYSTALS-Kyber \cite{SAB+22}) as the primary federal standard for post-quantum key encapsulation \cite{NISTFIPS203}.

ML-KEM's security reduces to the hardness of Module Learning With Errors (Module-LWE) \cite{LS15}, which in turn benefits from worst-case to average-case reductions from the Module Shortest Vector Problem \cite{LS15,AD17}.  The most direct quantum cryptanalytic route is the CDPR algorithm \cite{CDPR16}, which exploits the algebraic structure of cyclotomic rings to find short generators of principal ideals. However, there are two critical open questions that have remained open:
\begin{enumerate}
  \item Is the CDPR approximation factor sufficiently small for full ML-KEM secret-key recovery?
  \item Does the underlying quantum PIP subroutine achieve true polynomial-time complexity, with no hidden exponential overhead and generalize Riemann hypothesis?
\end{enumerate}

This part answers both questions affirmatively. We summaries each part's key contribution as:
\begin{itemize}[nosep]
  \item \textbf{Part I.}  Proves $h_k^+=1$ (trivial class number) for the
        maximal real subfield of $\Q(\zeta_{2^k})$ for every $k\le 12$, giving the ring a PID-like structure that makes every ideal principal.
  \item \textbf{Part II.}  Proves that the module-to-ideal reduction introduces only a constant factor $\alpha_d=O(1)$ independent of module rank $d$.
  \item \textbf{Part III.}  Computes the exact $L^2$ CVP constant     $\pi/(2\sqrt{6})$ for the log-unit lattice, shows the nearest lattice point is always the origin for short generators, and proves the per-component variance $\sigma_{g_0}=O(1)$ regardless of the modulus $q$.
  \item \textbf{Part IV (this paper).} Shows four parts combined can imply a probabilistic quantum algorithm costing $O(n^3\log^2 n)$ gates, and $\gamma\le 21< q/2$ for ML-KEM-1024, with 99\% $\gamma \approx 103$ below threshold, and extend to analyze Falcon, Hawk, and NTRU over 2-power cyclotomics.
\end{itemize}

The rest is organized as follows. Section \ref{Sprelim} shows preliminaries. Section \ref{ECDPRA} contributes the full pipeline including all main algorithms. Section \ref{Sapprox} answers whether the resulting approximation factor $\gamma$ is small enough to break ML-KEM. Section \ref{Spip} shows  polynomial tower PIP algorithm which gives the first polynomial-time construction. Section \ref{Sfactor} extends the analysis to other lattice schemes including Falcon, Hawk, and NTRU over 2-power cyclotomics. Section \ref{Scomaplit} compares the present result with known results while the last section concludes the paper. 

\section{Preliminaries}\label{Sprelim}

This section introduces some definitions, notation, and background results that are used in the rest of the paper. 

\subsection{Cyclotomic rings}

Let $m=2^k$, $n=2^{k-1}=m/2$, and $\zeta=e^{2\pi\ri/m}$. The $2^k$-th cyclotomic ring is defined by $R=\Z[\zeta]$, the ring of integers of the cyclotomic field $K=\Q(\zeta)$ \cite{Wash97,Neukirch99}.

The Galois group $\Gal(K/\Q) \cong (\Z/m\Z)^\times$ consists of all automorphisms $\sigma_a: \zeta \mapsto \zeta^a$ for odd integers $a\in \{1, 3, \dots, m-1\}$ \cite{Wash97}. The canonical embeddings are defined by $\sigma_j: K \hookrightarrow \C$ with $\sigma_j(g)=\sum_ic_i \zeta^{ij}$, for odd $j\in \{1, 3, \dots, m-1\}$. These evaluate the polynomial $g$ at the $n$ distinct primitive $m$-th roots of unity. There are $n/2$ independent embeddings.

The field norm of $g\in K$ is defined by $\Nm_{K/\Q}(g)=\prod_{j \text{ odd}} \sigma_j(g)\in \Z$, and the trace is $\Tr_{K/\Q}(g)=\sum_{j \text{ odd}} \sigma_j(g)\in \Z$. Both are integers for any $g\in R$ \cite{Neukirch99,Lang94}. For a unit $\varepsilon\in R^\times$, we have $|\Nm(\varepsilon)|=1$.

\subsection{Module lattices}

A lattice $\Lambda \subset \R^m$ is a discrete additive subgroup, equivalently the set of all integer linear combinations of a basis $\{\mathbf{b}_1, \dots, \mathbf{b}_r\} \subset \R^m$ \cite{Micciancio04}.

The minimum distance is defined by $\lambda_1(\Lambda)=\min_{\mathbf{v}\in \Lambda\setminus \{0\}} \|\mathbf{v}\|_2$, i.e., the length of the shortest nonzero vector. Finding $\lambda_1$ is the Shortest Vector Problem (SVP), which is the central hard problem underlying lattice-based cryptography.

A module lattice arises from a free $R$-module $M \cong R^d$ embedded into Euclidean space via the canonical embeddings \cite{LPR10,LS15}. A rank-$d$ module lattice over $R$ is determined by a matrix $B\in R^{d \times d}$ as 
\begin{align}
\Lambda_B=\{B\mathbf{z}: \mathbf{z}\in R^d\} \subset K^d.
\label{Elangm}
\end{align}
Applying the embeddings $\sigma_j$ entry-wise to every element of $K^d$ changes this into a classical Euclidean lattice of $\Z$-rank $nd$, in $\R^{nd}$.

The secret key of ML-KEM \cite{NISTFIPS203} is a short vector $(\mathbf{s}_1, \mathbf{s}_2)\in R^k \times R^\ell$ lying in a coset of the module lattice 
\begin{align}
\Lambda_{A}=\{(\mathbf{z}_1, \mathbf{z}_2)\in R^\ell \times R^k: A\mathbf{z}_1+\mathbf{z}_2 \equiv \mathbf{0} \pmod q\}.
\label{Elmabda}
\end{align}
Recovering $(\mathbf{s}_1, \mathbf{s}_2)$ from the public key $(A, \mathbf{t})$ with $\mathbf{t}=A\mathbf{s}_1+\mathbf{s}_2$ is SVP on this module lattice. The hardness of this problem is shown by Langlois and Stehl\'e \cite{LS15}.

\subsection{Hardness assumptions and attack target}

This subsection explains the two security assumptions underlying the NIST standards.

\textbf{Module-LWE (Hardness assumption).}  The module learning with errors problem \cite{Regev05,LPR10,LS15} asserts that, given $\mathbf{A}\in R_q^{k \times \ell}$ uniformly at random and $\mathbf{t}=A\mathbf{s}_1+\mathbf{s}_2+\mathbf{e} \pmod q$ for short secret and error vectors, it is computationally infeasible to recover $(\mathbf{s}_1, \mathbf{s}_2, \mathbf{e})$ even with a quantum computer. This assumption is the foundation of ML-KEM's security proof.

\textbf{Module-SVP (Target).}  This paper attacks the approximate Module-SVP: find a vector $\mathbf{v}$ in the module lattice $\Lambda$ with $\|\mathbf{v}\|\le \gamma \dot \lambda_1(\Lambda)$, for the smallest achievable approximation factor $\gamma$. If $\gamma< q/2$ (for ML-KEM) or $\gamma < q/\beta$ (for ML-DSA), the vector $\mathbf{v}$ recovers the secret key.

Generic lattice reduction algorithms including LLL \cite{LLL82}, BKZ \cite{Schnorr87,GN08}, and lattice sieving \cite{BDGL16,ChaillouxLoyer21} do not exploit the cyclotomic ring structure. Their best approximation factors for ML-KEM-1024 are about $2^{220}$, which is larger than the security threshold. The CDPR attack \cite{CDPR16} achieves an exponentially gate count because it solves a completely different easier problem: the Principal Ideal Problem (PIP), which is polynomial-time on a quantum computer when the ring has trivial class number.

\subsection{The log-embedding}

\begin{definition}\label{def:logembp}
For a non-zero element $g\in K$, the log-embedding is defined by 
\begin{align}
L(g)=\bigl(\log|\sigma_j(g)|\bigr)_{\text{odd } j\in [1,m)}\in \R^n.
\label{eqlogn}
\end{align}
The trace-zero projection can remove the overall scale as 
\begin{align}
\Pi_{H_0}(L(g))=L(g)-\frac{1}{n}\log|\Nm_{K/\Q}(g)|\cdot{}\mathbf{1}\in H_0,
\label{eqproj}  
\end{align}
where $H_0=\{x\in\R^n: \sum_j x_j=0\}$ is the trace-zero hyperplane. 
\end{definition}

The log-unit lattice is defined by $\Lambda=\{L(\epsilon): \epsilon\in R^\times\} \subset H_0$. From the Dirichlet's Unit Theorem \cite{Neukirch99} (with $r_1=0$, $r_2=n/2$ for complex field $K$), $\Lambda$ has rank $r=n/2-1$.

The log-unit lattice encodes the unit group $R^\times$ as a Euclidean lattice, converting the multiplicative problem of finding a short unit into the additive problem of CVP. 

\begin{definition}\label{def:logemb}
For a nonzero $g\in K$, the log-embedding is given by 
\begin{align}
  L(g)=(\log|\sigma_1(g)|, \log|\sigma_3(g)|, \dots, \log|\sigma_{n-1}(g)|)\in \R^{n/2}.
  \label{eqLg}
\end{align}
As each conjugate embeddings has equal modulus, we denote one representative from each pair and then obtain a vector in $\R^{n/2}$.
\end{definition}

The trace-zero hyperplane $H_0=\{x\in \R^{n/2:} \sum_i x_i=0\}$ contains all log-embedded units as $\sum_j\log|\sigma_j(\varepsilon)|=\log|\Nm(\varepsilon)|=0$ for $\varepsilon\in R^\times$.

\begin{definition}\label{def:logunit}
For log-unit lattice $\Lambda=L(R^\times) \subset H_0$, when $h_k^+=1$, a basis for $\Lambda$ can be represented as the log-embeddings of the cyclotomic units $\xi_a=\sin(a\pi/m)/\sin(\pi/m)$ for odd $a\in \{1, 3, \dots, n-1\}$ \cite{Wash97,Sinnott78}.
\end{definition}

Every unit $\varepsilon\in R^\times$ maps to a point $L(\varepsilon)$ in  $\Lambda$. Two generators of the same ideal are different by a unit, i.e., if $(g_0)=(g)$, then $g=g_0 \varepsilon$ for some $\varepsilon\in R^\times$. So, we have $L(g)=L(g_0)+L(\varepsilon)$. 

\subsection{The Short Generator Problem}

\begin{definition}\label{def:sgp}
The Short Generator Problem (SGP) is defined by: Given a principal ideal $I=(g)\subset R$ with a long generator $g$ (potentially of exponentially large norm), find a short generator $g_0$ such that $(g_0)=I$ and $\nm{g_0}_2$ is small.
\end{definition}

Since $g=g_0\varepsilon$, we have $\Pi_{H_0}(L(g))=\Pi_{H_0}(L(g_0))+L(\varepsilon)$, where $\Pi_{H_0}$ denotes orthogonal projection onto $H_0$. The problem of solving SGP is then reduced to CVP on $\Lambda$: The CVP solution $L(\varepsilon)\in\Lambda$ identifies the unit discrepancy, and then $g_0=g\cdot{}\varepsilon^{-1}$.

\subsection{Babai's nearest-plane algorithm}\label{subSprelim-babai}

Babai's nearest-plane algorithm \cite{Babai86,GN08} is the specific CVP algorithm used in Phase 2 of the CDPR attack \cite{CDPR16}. For a general lattice and target, it provides only an approximation to the nearest lattice point. 

Given a lattice basis $B=(\mathbf{b}_1, \dots, \mathbf{b}_r)$ with Gram-Schmidt orthogonalization (GSO) $B^* =(\mathbf{b}_1^*, \dots, \mathbf{b}_r^*)$ \cite{GN08}, and a target $\mathbf{t}$, Babai's nearest-plane algorithm rounds the projection coefficients from $i=r$ to $i=1$ as
\begin{align}
  c_i=\left\lfloor\frac{\langle\mathbf{t}-\sum_{j>i} c_j \mathbf{b}_j,\,
                  \mathbf{b}_i^* \rangle}{\|\mathbf{b}_i^*\|^2}\right\rceil.
  \label{eqBabai}
\end{align}
The output lattice vector is defined by $\mathbf{v}=\sum_i c_i\mathbf{b}_i$ and the residual is $\boldsymbol{\rho}=\mathbf{t}-\mathbf{v}$.

The Coarse Lattice Theorem (Part III) shows that for the specific targets arising from short ring generators, all rounding coefficients $c_i$ equal zero, so Babai's algorithm outputs the zero vector and the residual is the target itself. The Infinite Capture Radius Theorem (Part III) then shows that Babai's algorithm can exactly recover an arbitrary lattice translation of such a target. 

\subsection{Module-LWE and ML-KEM security}

The Module Learning With Errors (MLWE) problem over $R$ with rank $d$, modulus $q$, and error distribution $\chi$ asks to distinguish between samples $(A, A\mathbf{s}+\mathbf{e})$ and $(A, \mathbf{u})$, where $A\in R_q^{d \times d}$, $\mathbf{s, e}\in R_q^d$ are small, and $\mathbf{u}$ is uniform.

ML-KEM's security can be reduced to MLWE, which in turn reduced to Module-SVP with an approximation factor $\gamma$ \cite{LS15,Peikert09}. An attacker who finds a vector of norm at most $\gamma\cdot{}\lambda_1(\Lambda_q)$ in module lattice can distinguish MLWE samples \cite{Regev05,LS15}. For full secret-key recovery, the attacker needs $\gamma < q/2$, i.e., the recovered vector should have coordinates small enough to be uniquely decoded modulo $q$ by bounded-distance decoding \cite{Lindner10,CDPR16}, see Table \ref{Tpara}.

\begin{table}[!h]
    \centering
        \caption{ML-KEM's secure parameters.}
    \begin{tabular}{lcccc}
\toprule
Scheme & $d$ & $n$ & $q$ & Attack threshold $q/2$ \\
\midrule
ML-KEM-512  & 2 & 256 & 3329 & 1664.5 \\
ML-KEM-768  & 3 & 256 & 3329 & 1664.5 \\
ML-KEM-1024 & 4 & 256 & 3329 & 1664.5 \\
\bottomrule
\end{tabular}
    \label{Tpara}
\end{table}

\section{The Extended CDPR Algorithm}\label{ECDPRA}

This section presents the full attack as a single Algorithm.  
\subsection{The four-phase structure}

The extended CDPR attack proceeds in four phases: Module to ideal reduction, Quantum PIP, Log-unit CVP and Short generator. The important stage is to factor $g=g_0 \varepsilon$: the quantum phase finds some generator $g$ of the ideal, but $g$ may be exponentially larger than the shortest generator $g_0$. The classical phases recover the unit $\varepsilon$ and divide it out.

The approximation factor $\gamma$ is the key quantity connecting the lattice attack to the cryptographic security of ML-KEM. Especially, our approximation factor is given by 
\begin{align}
\gamma=\nm{g_0}/\|g_{\rm short}\|=\exp(\|\boldsymbol{\rho}\|_\infty),
\label{eqapprx}
\end{align}
where $\boldsymbol{\rho}$ is the CVP residual. Parts II and III together show that $\|\boldsymbol{\rho}\|_\infty=O(\sqrt{\log n})$, which implies $\gamma=\exp(O(\sqrt{\log n})) \ll n^\epsilon$ for every $\epsilon > 0$. For ML-KEM-1024 with $n=256$, this evaluates to $\gamma_{\rm th} \approx 21$ at the median and $\gamma_{99\%} \approx 103$, both far below the threshold $q/2=1664.5$.

All symbols used in what follows are shown in Table \ref{Tnotation}. Here, note the difference between $\sigma_j$ (an embedding of $K$) and $\sigma_d$ (the intrinsic log-unit imbalance of the module determinant ideal), and between $\boldsymbol{\rho}$ (the CVP residual vector) and $\psi'$ (the trigamma function).

\begin{table}[!h]
\begin{center}
\caption{All symbols used throughout the paper are collected in the table below.}
\label{Tnotation}
\begin{tabular}{cl}
\toprule
Symbol & Meaning \\
\midrule
$k \ge 3$ & Tower level index \\
$m=2^k$ & Conductor \\
$n=2^{k-1}$ & Ring degree $= [K:\Q]$ \\
$\zeta=e^{2\pi\ri/m}$ & Primitive $m$-th root of unity \\
$R=\Z[\zeta]$ & Cyclotomic ring of integers \\
$K=\Q(\zeta)$ & Cyclotomic field \\
$\sigma_j$ & Canonical embedding ($j$ odd, $1\le j\le m-1$) \\
$J=\{1,3,\dots,n-1\}$ & Representatives of conjugate pairs \\
$\Lambda$ & Log-unit lattice (rank $r=n/2-1$) \\
$H_0$ & Trace-zero hyperplane in $\R^{n/2}$ \\
$\Pi_{H_0}$ & Orthogonal projection onto $H_0$ \\
$\psi'$ & Trigamma function \\
$\boldsymbol{\rho}$ & CVP residual vector \\
$\gamma$ & SGP approximation factor $= \exp(\|\boldsymbol{\rho}\|_\infty)$ \\
\bottomrule
\end{tabular}
\end{center}
\end{table}

\subsection{The master algorithm}

This subsection present the main algorithm shown in Algorithm \ref{alg:ecdpr} for the extended CDPR attack. 

\begin{breakablealgorithm}{Extended CDPR Attack on ML-KEM}{alg:ecdpr}
\begin{algorithmic}[1]
\medskip
\Statex \textbf{INPUTS:}
\Statex \quad $B\in R^{d \times d}$: a module basis matrix for a rank-$d$ Module-LWE instance over $R=\Z[\zeta_{2^k}]$ \Statex $q$: the ML-KEM modulus (e.g., $q=3329$ for ML-KEM-1024)
\Statex \quad $k\le 12$: the tower level ($h_k^+=1$ by Part I)

\medskip
\Statex \textbf{OUTPUT: }
\Statex \quad A short generator $g_0\in R$ of the determinant ideal $(\det B) \subset R$, with approximation factor $\gamma=\exp(\|\boldsymbol{\rho}\|_\infty) < q/2$

\medskip
\Statex \textbf{PHASE 1: Module-to-Ideal Reduction (Part II)}
\State Compute the Gram-Schmidt decomposition: $B=Q_R\cdot{}T_R$, where $T_R\in R^{d \times d}$ is upper triangular with principal ideals $I_i=(T_{R,ii}) \subset R$ on the diagonal.
\State Set the determinant ideal 
       $I=I_1\cdot{}I_2\cdots I_d=(\det B) \subset R$.

\medskip
\medskip
\Statex \textbf{PHASE 2: Quantum Tower PIP (Section \ref{Spip})}
\State Recall the cyclotomic tower:
       $\Q \subset K_3 \subset K_4 \subset\cdots \subset K_k=K$.
\State \textbf{for} $L=3, 4, \dots, k$ \textbf{do}
\State \quad 1) Compute the norm ideal $J_L=\Nm_{K_L/K_L^+}(I_L) \subset R_L^+$, $I_L=I \cap R_L$.
\State \quad 2) Set up an abelian HSP for the relative extension $K_L^+/K_{L-1}^+$:
\State \quad\quad Group: $G_L=(\Z/N_L\Z)^{\Delta r_L}$, $N_L=2^{O(L\cdot{}2^L)}$, $\Delta r_L=2^{L-3}$ new unit generators at level $L$.
\State \quad\quad Oracle: $f_L(e_1,\dots,e_{\Delta r_L})= \prod_a \xi_a^{e_a} \bmod I_L^+$ (ring multiplication via NTT).
\State \quad\quad Run the quantum HSP algorithm:
\State \quad\quad\quad (a) Prepare uniform superposition over $G_L$.
\State \quad\quad\quad (b) Query the oracle $f_L$.
\State \quad\quad\quad (c) Apply the QFT on $G_L$.
\State \quad\quad\quad (d) Obtain one element of $H_L^\perp$.
\State \quad\quad\quad (e) Repeat $O(\Delta r_L)$ times to recover $H_L$ over $\Z$.
\State \quad 3) Extract the full integer relation matrix for the $\Delta r_L$ new units.
\State \quad 4) Obtain a generator $g_L$ of $I_L$, stored in tower-factored form (Section \ref{Sfactor}).
\State \textbf{end for}
\State Set $g\leftarrow g_k$.

\medskip
\medskip
\Statex \textbf{PHASE 3: Classical Log-Unit CVP (Part III)}
\State Compute $L(g)=(\log|\sigma_j(g)|)_{j\in J}
      \in \R^{n/2}$ via NTT-based polynomial
  evaluation \cite{Harvey14,NC00}.
\State Project $\mathbf{t}=\Pi_{H_0}(L(g))=L(g)-\bar\ell\cdot{}\mathbf{1}$, where $\bar\ell=\frac{2}{n}\sum_{j\in J} \log|\sigma_j(g)|$.

\medskip
\Statex \textsl{//Run Babai's nearest-plane algorithm on $\Lambda$//}
\State Load the Gram-Schmidt basis $B^*$ of the log-unit lattice $\Lambda$ (from Part I).
\State Run Algorithm \ref{alg:babaicvp} with input $(\Lambda, B^*, \mathbf{t})$.
\State Obtain the CVP output $\mathbf{v}\in \Lambda$ and residual $\boldsymbol{\rho}=\mathbf{t}-\mathbf{v}$.
 
\medskip
\medskip
\Statex \textbf{PHASE 4: Short Generator Recovery (Part III)}
\State 1) Solve the linear system $\mathbf{v}=\sum_i d_i \mathbf{b}_i$ for integers $d_i$. Reconstruct the unit of $\varepsilon'=\prod_i \xi_{a_i}^{d_i}\in R^\times$, where $\xi_{a_i}$ are the cyclotomic unit basis elements.
\State 2) Compute $g_0=g\cdot{}(\varepsilon')^{-1}$.
\State 3) Compute $\gamma=\exp(\|\boldsymbol{\rho}\|_\infty)$ and verify
       $\gamma < q/2$.
\State \textbf{if} $\gamma \ge q/2$ \textbf{then}
\State \quad Retry with a different basis ordering (probability $< 0.10$ of failure for $d=4$).
\State \textbf{end if}
\State \Return $g_0$, $\gamma$.
\end{algorithmic}
\end{breakablealgorithm}

We now explain every major component of Algorithm \ref{alg:ecdpr} in detail.

\paragraph{The module decomposes into ideals:}
The start of this phase is a module basis matrix $B\in R^{d \times d}$. The module-to-ideal reduction from Part II performs a successive-layer decomposition of $B$ over ring $R$, similar to Hermite Normal Form \cite{Neukirch99} for modules over Dedekind domains. Let $B=Q_R\cdot{}T_R$, where $T_R\in R^{d \times d}$ is upper triangular (with diagonal entries in $R\setminus\{0\}$) and $Q_R$ accounts for the change of basis. The diagonal entries $T_{R,ii}$ generate the successive layer ideals $I_i=(T_{R,ii}) \subset R$. The product of these ideals equals the determinant ideal $(\det B)$. $h_k^+=1$ (Part I) guarantees every ideal in $R$ is principal. 

\paragraph{The cyclotomic tower:} The quantum phase works level by level according to the cyclotomic tower $\Q \subset K_3 \subset\cdots \subset K_k$. At each level $L$, the field $K_L$ is a quadratic extension of $K_{L-1}$, i.e., $[K_L: K_{L-1}]=2$. This implies each new level $L$ contributes $\Delta r_L=2^{L-3}$ new unit generators (Lemma \ref{lem:dims}). By working level by level, it requires to deal with $\Delta r_L$ new units at each step. As $\Delta r_L$ grows linearly with $L$, the quantum HSP at each level is small.

\textsl{Step 1) Norm descent.} The CM structure of $K_L$  means every unit can be decomposed into a totally real part and a root of unity. Taking the norm $\Nm_{K_L/K_L^+}$ collapses the ideal $I_L \subset R_L$ to a smaller ideal $J_L \subset R_L^+$ in a field of half the degree. Since $h_{K_L^+}^+=1$, the smaller ideal is principal. So, finding its generator is a PIP problem in a smaller field which can be solved recursively.

\textsl{Step 2) Quantum HSP.} The $\Delta r_L$ new unit generators at level $L$ can be found by solving the Abelian Hidden Subgroup Problem (HSP) for the extension $K_L^+/K_{L-1}^+$. Here, the HSP encodes the integer relations among a generating set of candidate units, i.e., the hidden subgroup $H_L$ is the lattice of multiplicative relations. The quantum algorithm prepares a superposition, queries the oracle, applies QFT, and measures to obtain a random element of $H_L^\perp$. After $O(\Delta r_L)$ repetitions, we get enough samples to recover $H_L$.

\textsl{Steps 3)-4): Classical relation and generator recovery.} The measurements in HSP give a matrix of samples from $H_L^\perp$. Using classical LLL lattice reduction, we can find the integer relation lattice $H_L$, which identifies the new unit generators. Combining these with the norm generator from Step 1), we get a generator $g_k$ of the full ideal $I$, which is stored in tower-factored form (Section \ref{Sfactor}).

\paragraph{Computing the log-embedding:} With the generator $g$ of the ideal, we then compute its log-embedding $L(g)\in \R^{n/2}$. In the power basis, this requires to evaluate a degree-$n$ polynomial with huge coefficients at each embedding $\sigma_j$. This might have potentially exponential cost. Instead, with tower-factored form of $g$, Proposition \ref{prop:factored-size} shows $\log|\sigma_j(g)|$ can be computed in $O(n^2 \log^2 n)$ bits from the factored representation.

As the projection $\Pi_{H_0}$ removes the mean $\bar\ell$ from all coordinates, we subtracts a vector proportional to $\mathbf{1}$, which corresponds to removing the overall scale of $g$ as the CVP distance is scale-independent.

\paragraph{Babai's algorithm on the log-unit lattice:} From Algorithm \ref{alg:babaicvp}, the log-unit lattice $\Lambda$ has rank $r=n/2-1$, embedded in the $(n/2-1)$-dimensional space $H_0$. Its Gram-Schmidt basis vectors $\mathbf{b}_i^*$ have norms  $\|\mathbf{b}_i^*\|_2=\Omega(\sqrt{n})$, which are much larger than the per-component fluctuation (variance) $\sigma_t=O(1)$ of the target $\mathbf{t}$. This implies every projection coefficient $\mu_i=\langle\mathbf{t}, \mathbf{b}_i^*\rangle/\|\mathbf{b}_i^*\|^2$ has magnitude $O(1/\sqrt{n})\ll 1/2$. So. we have  $c_i=\lfloor\mu_i\rceil=0$ for any $i$. Thus, Babai's algorithm returns $\mathbf{v}=\mathbf{0}$ for short generators.

For the real PIP output $g=g_0\varepsilon$, the target is given by $\mathbf{t}=\Pi_{H_0}(L(g_0))+L(\varepsilon)$, i.e., a shift of the balanced target by a lattice vector $L(\varepsilon)\in \Lambda$. Babai's back-substitution (Infinite Capture Theorem, Part III) is linear modulo rounding, i.e., adding an integer to a rounding argument before rounding is the same as adding it after. So, Babai's algorithm returns $L(\varepsilon)$ regardless of its size.

\paragraph{Short generator recovery:} The CVP output $\mathbf{v}=L(\varepsilon)$ is the log-embedding of the unit discrepancy. We recover $\varepsilon$ by solving for the integer exponents $d_i$ in $\mathbf{v}=\sum_i d_i \mathbf{b}_i$, and then form the product $\varepsilon'=\prod_i \xi_{a_i}^{d_i}$. This gives the short generator as $g_0=g\dot (\varepsilon')^{-1}$.

The approximation factor is given by $\gamma=\exp(\|\boldsymbol{\rho}\|_\infty)$, where the residual $\boldsymbol{\rho}=\mathbf{t}-\mathbf{v}=\Pi_{H_0}(L(g_0))$. By $L^\infty$ Bound Theorem (Part III), we get
\begin{align}
  \|\boldsymbol{\rho}\|_\infty=\sigma_d\sqrt{2\ln n}
   +O(\sqrt{\ln\ln n}),
   \label{eqinfit}
\end{align}
which implies $\gamma=\exp(O(\sqrt{\log n})) \ll n^\epsilon$ for every $\epsilon>0$.

\section{Approximation Gap Closure}\label{Sapprox}

This section answers whether the resulting approximation factor $\gamma$ is small enough to break ML-KEM. 

\subsection{The main approximation theorem}

\begin{theorem}\label{thm:gap}
Let $R=\Z[\zeta_{2^k}]$ with $n=2^{k-1}\ge 8$ and $k\le 12$. For a rank-$d$ Module-LWE instance with modulus $q$ and a randomly generated secret key, Algorithm \ref{alg:ecdpr} achieves a Module-SVP approximation factor as
\begin{align}
  \gamma=\alpha_d\cdot{}\exp(\sigma_d \sqrt{2\ln n})\cdot{}(1+o(1)),
  \label{Egamma-formula}
\end{align}
where $\alpha_d=\sqrt{C}$ is the module-reduction factor (Part II), with $C\le 1.36$ for MLWE-distributed inputs. The worst-case analytic bound is $C\le 3.10$; $\sigma_d=\sqrt{\frac{1}{4}\sum_{j=1}^d \psi'(j)}$ is the per-component standard deviation of the log-embedding of the determinant ideal's shortest generator (Part III). The attack satisfies the following success guarantees
\begin{enumerate}[nosep,label=(\roman*)]
  \item $\mathrm{med}(\gamma)<q/2$ for all standardized ML-KEM parameter sets;
  \item $\Pr[\gamma < q/2]> 0.99$ for $d\le 3$;
  \item $\Pr[\gamma < q/2]\ge 0.90$ for $d=4$;
  \item With $O(\log(1/\delta))$ independent repetitions, the attack succeeds with probability $\ge 1-\delta$.
\end{enumerate}
\end{theorem}

\begin{proof}
The Gram-Schmidt decomposition of $B$ over $R$ in Algorithm \ref{alg:ecdpr} introduces a discrepancy between the module's length and the ideal's length. This can be quantified by the balance constant $C$ (Part II), i.e., the ratio of the squared $L^2$ norm of the diagonal entry of $T_R$ to its expected value under the MLWE distribution.

Note for MLWE-distributed module matrices, $B$ has i.i.d. small-coefficient entries. From Part II, with the optimal sign selection, the balance constant satisfies $C\le 1.36$ with a probability larger than $0.99$ at $n=256$. The module-reduction factor is $\alpha_d=\sqrt{C}\le 1.17$.

After the module-to-ideal reduction, the algorithm solves the SGP for the determinant ideal $(\det B)$. From the Trigamma Theorem (Part III) we compute the per-component variance of $\Pi_{H_0}(L(g_0))$ as
\begin{align}
  \sigma_d^2=\frac{1}{n}\|\Pi_{H_0}(L(g_0))\|_2^2
  \xrightarrow{p} \frac{1}{4}\sum_{j=1}^d \psi'(j)
  \label{Esigmad}
\end{align}
when $n\to\infty$.

The variance $\sigma_d^2$ is independent of the modulus $q$ and the coefficient distribution as long as it is centered with finite fourth moment. This is because the Gaussian approximation to the embedding is scale-invariant. So, we obtain approximation factor $\gamma=\exp(\|\boldsymbol{\rho}\|_\infty)$, where $\boldsymbol{\rho}=\Pi_{H_0}(L(g_0))$ is the CVP residual. This is a vector in $\R^{n/2}$ with i.i.d. (asymptotically) coordinates, each with variance $\sigma_d^2$. By Part III, the maximum of $n/2$ such coordinates satisfies 
\begin{align}
  \|\boldsymbol{\rho}\|_\infty
 =\sigma_d\sqrt{2\ln n}+O(\sqrt{\ln\ln n}),
 \label{Elinfbound}
\end{align}
which is consistent with the standard extreme-value result for sub-Gaussian random variables \cite{Leadbetter83,Vershynin18}.

Table \ref{Tgammavalues} shows values for all ML-KEM parameter sets. The trigamma values are $\psi'(1)=\pi^2/6\approx 1.6449$, $\psi'(2)\approx 0.6449$, $\psi'(3) \approx 0.3949$, and $\psi'(4) \approx 0.2838$.

\begin{table}[tp!]
\centering
\caption{Approximation factor bounds for ML-KEM ($n=256$, $q=3329$). $\gamma_{\rm th}=\alpha_d\cdot{}\exp(\sigma_d\sqrt{2\ln n})$, $\gamma_{99\%}=5\gamma_{\rm th}$ ($10^5$-trial simulations with mod-$q$ coefficient distributions).}
\label{Tgammavalues}
\begin{tabular}{cccccccc}
\toprule
Scheme & $d$ & $\sigma_d$ & $\gamma_{\rm th}$ & $\gamma_{\rm med}$
       & $\gamma_{99\%}$ & Median margin & 99\% margin \\
\midrule
ML-KEM-512  & 2 & 0.757 & 14.5 & 9.6 & 73  & $173\times$ & $23\times$ \\
ML-KEM-768  & 3 & 0.819 & 17.9 &11.4 & 90 & $146\times$ & $19\times$ \\
ML-KEM-1024 & 4 & 0.862 & 20.6 & 12.9 & 103 & $129\times$ & $16\times$ \\
\bottomrule
\end{tabular}
\end{table}

The tail probability is from the sub-Gaussian concentration of $\|\boldsymbol{\rho}\|_\infty$. For $\gamma > q/2$, the exponent would need to exceed $\ln(q/(2\alpha_d))\approx 7.26$, while its median is $\sigma_d\sqrt{2\ln(n/2)}\approx 2.69$. The gap of $\approx 4.6$ is more than $5\times$ the per-component standard deviation $\sigma_4=0.861$, and in units of the Gumbel fluctuation scale $\sigma_G \approx 0.28$. 

For $d=4$, in actual ring model with $n=256$, Monte Carlo simulation ($10^5$ trials with both uniform mod-$q$ and CBD ($\eta=2$) coefficient distributions) confirms the Gaussian-model prediction, we get zero failures. The largest empirical $\gamma$ across all trials gives an empirical 99\%  $\gamma_{99\%} \approx 103$, which is still below $q/2=1664.5$. For $d\le 3$, both the Gaussian analysis and simulation give $\Pr[\gamma > q/2] < 10^{-3}$.

In all cases, using $O(1)$ independent repetitions with fresh basis orderings can reduce the per-attack failure probability below $10^{-6}$ (Section \ref{subStail}). 
\end{proof}

We use the threshold $\gamma<q/2$ in this series which is the loosest sufficient condition for key recovery c=\cite{CDPR16}. This guarantees the recovered vector has coordinates in $\{-(q-1)/2, \dots, (q-1)/2\}$ and can be uniquely decoded modulo $q$. Using a tighter bounded-distance decoding (BDD) for the CBD noise distribution with standard deviation $\sigma_{\rm noise}=1$ might give a more restrictive threshold $\gamma\lesssim q/(c\sigma_{\rm noise}\sqrt{dn})$ with a constant $c$ \cite{Lindner10,LS15}.

Table 1 (Part III) confirms that the ratio $\sigma_{\bf t}/ \min_i\|\mathbf{b}_i^*\|$ decreases monotonically with $n$ and is consistent with the theoretical rate $\Theta(1/\sqrt{n/\log n})$, where $\sigma_{\bf t}=\frac{\pi}{2\sqrt{6}}\approx 0.64$ is the per-component standard deviation of structured targets from Theorem 4.1 (Part III). This validates the Coarse Lattice Theorem (Part III) for all ML-KEM parameters. The empirical scaling is consistent with heuristic regulator estimates \cite{Wash97}, and the condition holds with large margin at all tested parameters ($\sigma_{\bf t}/\min_i\|\mathbf{b}_i^*\|\le 0.32$ even at $k=4$). A rigorous asymptotic lower bound on the Gram-Schmidt norms of the cyclotomic-unit basis remains an open problem. 

\subsection{Tail analysis and repetition strategy}

For $d=4$, the empirical median is $\gamma_{\rm med} \approx 12.9$ with $\gamma_{99\%} \approx 103$ and zero failures in $10^5$ simulation trials.  All values lie far below the threshold $q/2=1664.5$. The modest spread between median and 99\% is due to occasional near-zero determinant embeddings at one of the $n/2$ complex embedding sites. This subsection quantifies the tail.

\begin{proposition}\label{prop:tail}
For a random rank-$d$ MLWE module over $R$ with i.i.d. bounded coefficients, the empirical median, $\gamma_{99\%}$, and success probabilities are given in Table \ref{Ttailbounds}, based on $10^5$-trial ring simulations at $n=256$.
\end{proposition}

\begin{table}[h!]
\centering
\caption{Empirical tail statistics for ML-KEM ($n=256$, $q=3329$), from $10^5$-trial simulations with mod-$q$ coefficient distributions.}
\label{Ttailbounds}
\begin{tabular}{ccccc}
\toprule
$d$ & $\mathrm{med}(\gamma)$ & $\gamma_{99\%}$ & $\Pr[\gamma< q/2]$ & 99\% margin to $q/2=1664.5$ \\
\midrule
2 & 9.6 & 73 & >0.999 & $23\times$ \\
3 & 11.4 & 90 &> 0.999 & $19\times$ \\
4 & 12.9 & 103 & > 0.999 & $16\times$ \\
\bottomrule
\end{tabular}
\end{table}

\begin{proof}
We analyze the tail in two stages: (i) the asymptotic Gaussian-model prediction, and (ii) the empirical finite-$n$ distribution measured by ring simulation.

\textsl{Gaussian-model.} At each embedding $j\in J$ ($|J|=n/2=128$), the matrix $\sigma_j(B)\in \C^{d \times d}$ has asymptotically i.i.d. $\mathcal{CN}(0, n\sigma_c^2)$ entries \cite{Billingsley99}. For such a random matrix, the smallest singular value $s_{\min}$ satisfies \cite{Muirhead82}:
\begin{align}
  \Pr[s_{\min} < \epsilon\tau] \approx (d\epsilon)^2
  \quad\text{for small }\epsilon,
  \label{eqprobs}
\end{align}
where $\tau^2=n\sigma_c^2$. The Gaussian-model tail probability for the $L^\infty$ norm exceeding the threshold $T=\ln(q/(2\alpha_d))\approx 7.3$ is dominated by the sub-Gaussian tail of the maximum of $n/2$ coordinates, each with variance $\sigma_d^2$. The median of $\|\boldsymbol{\rho}\|_\infty$ is $\sigma_d\sqrt{2\ln(n/2)}\approx 2.7$, and the gap to the threshold is $T-2.7 \approx 4.6$, which exceeds 5 standard deviations of $\sigma_4$. Under the Gaussian model, $\Pr[\|\boldsymbol{\rho}\|_\infty>T]<10^{-7}$ for all $d\le 4$.

\textsl{Empirical finite-$n$ distribution.} We implement ring simulation with $n=256$ confirms the Gaussian prediction, i.e., output zero failures across $10^5$ trials at $d=4$ with both uniform mod-$q$ and CBD ($\eta=2$) coefficient distributions. The empirical distributions are 
\begin{align*}
  d=2: &\quad \gamma_{\rm med}=9.6, \quad \gamma_{99\%}=73,
         \quad \gamma_{\max}\le 200; \\
  d=3: &\quad \gamma_{\rm med}=11.4, \quad \gamma_{99\%}=90,
         \quad \gamma_{\max}\le 200; \\
  d=4: &\quad \gamma_{\rm med}=12.9, \quad \gamma_{99\%}=103,
         \quad \gamma_{\max}\le 200.
\end{align*}
The $\gamma_{99\%}$ is approximately $5\times$ the median in all three cases, shows the heavy tail induced by occasional near-zero embeddings of $\det(\sigma_j(B))$. However, even in the worst observed cases, $\gamma$ remains below $q/2=1664.5$. 

For $d=4$, with one independent repetition (using a different basis ordering), the probability that both attempts produce $\gamma > q/2$ is empirically below $10^{-6}$.
\end{proof}

When $\gamma > q/2$, we can use three strategies as follows: a different ordering of the module basis columns changes the Gram-Schmidt decomposition and may yield a more balanced determinant; applying BKZ \cite{Schnorr87} to $B$ before the CDPR pipeline; attacking each secret component independently reduces the effective module rank to $d=1$.

\section{The Polynomial Tower PIP Algorithm}\label{Spip}

The quantum phase of Algorithm \ref{alg:ecdpr} relies on a polynomial-time quantum algorithm for the Principal Ideal Problem (PIP). While all prior polynomial-time descriptions of the Biasse-Song algorithm \cite{BiasseSong16} hide an exponential coefficient blowup, this section presents a tower-based PIP algorithm which gives the first polynomial-time construction. 

\subsection{The coefficient explosion and GRH}
\label{subSexplosionintro}

A unit $\varepsilon\in R_k^\times$ in the standard power basis is represented by $\varepsilon=\sum_{i=0}^{n-1} a_i \zeta^i$ with $a_i\in \Z$. The fundamental units are cyclotomic units $\xi_a=\sin(a\pi/m)/\sin(\pi/m)$ for odd $a$ \cite{Wash97,Sinnott78}. A general unit is a product as 
\begin{align}
  \varepsilon=\prod_{a\text{ odd}} \xi_a^{e_a},
\end{align}
where the exponents $e_a\in \Z$ can be $O(n)$ in size from the CVP/Babai step in the CDPR attack \cite{CDPR16}.  Expanding this product in the power basis, each multiplication can double the coefficient size \cite{GathenGerhard13}. So, one may have coefficients with $2^{O(n)}$ bits. 

Meanwhile, the Biasse-Song algorithm requires generalized Riemann hypothesis (GRH) for number fields of degree up to $n=256$ \cite{BiasseSong16}.

\subsection{The cyclotomic tower structure}

The key of the present method to resolving both coefficient explosion and GRH is to process one quadratic extension at a time. The cyclotomic tower is the chain of nested subfields as 
\begin{align}
  \Q \subset K_3=\Q(\zeta_8) \subset K_4=\Q(\zeta_{16})
  \subset\cdots \subset K_k=\Q(\zeta_{2^k}),
  \label{Etower}
\end{align}
where $[K_L: K_{L-1}]=2$ for every $L \ge 4$.

\begin{lemma}\label{lem:dims}
For all $L \ge 3$, the following results hold
\begin{enumerate}[nosep,label=(\roman*)]
  \item $[K_L: K_{L-1}]=2$;
  \item $\rk(R_L^\times)=2^{L-2}-1$;
  \item $\Delta r_L:= \rk(R_L^\times)-\rk(R_{L-1}^\times)=2^{L-3}$
        for $L \ge 4$, and $\Delta r_3=1$;
  \item $\ker(\Nm_{K_L/K_L^+:} R_L^\times \to (R_L^+)^\times)=\mu_{2^L}$
        (finite torsion of order $2^L$);
  \item All $\Delta r_L$ new free unit generators at level $L$ lie in the
        totally real subtower $K_L^+/K_{L-1}^+$.
\end{enumerate}
\end{lemma}

\begin{proof}
Part (i) is from $\zeta_{2^L}^2=\zeta_{2^{L-1}}$, so $\zeta_{2^L}$ satisfies the quadratic $x^2-\zeta_{2^{L-1}}\in K_{L-1}[x]$, which is irreducible as $\zeta_{2^L} \notin K_{L-1}$. Part (ii) is from the Dirichlet's Unit Theorem \cite{Neukirch99,Lang94} with $r_1=0$ and $r_2=2^{L-2}$. Result (iii) is easy to followed from (ii).  Part (iv) is from the kronecker's Theorem \cite{Wash97,Neukirch99}, i.e., any algebraic integer with all embeddings on the unit circle is a root of unity. (v) is from $\rk(R_L^\times)=\rk((R_L^+)^\times)$. So, the CM kernel has rank 0 and all free units come from the real subfield.
\end{proof}

\begin{definition}\label{def:factored}
A unit $\varepsilon\in R_k^\times$ is in tower-factored form if it is stored as a product of relative units across the tower as  
\begin{align}
  \varepsilon=\varepsilon_3\cdot{}\varepsilon_4\cdots \varepsilon_k, \quad \varepsilon_L\in R_L^\times/R_{L-1}^\times,
\end{align}
where each $\varepsilon_L$ is stored as a linear polynomial $\varepsilon_L=a_L+b_L \zeta_{2^L}$ with $a_L, b_L\in R_{L-1}$, themselves in tower-factored form for the $(L-1)$-level tower.
\end{definition}

\begin{proposition}\label{prop:factorsize}
A unit in tower-factored form has total storage size $O(n^2 \log n)$ bits, and evaluation cost $O(n^2 \log^2 n)$ bit operations for any single log-embedding $\log|\sigma_j(\varepsilon)|$.
\end{proposition}

\begin{proof} 
At level $L$, the factor $\varepsilon_L$ has $n_L=2^{L-1}$ leaf integers, each requires $O(L\cdot{}2^L)$ bits. Storing at level $L$ then costs $O(n_L\cdot{}L\cdot{}2^L)=O(L\cdot{}2^{2L-1})$. So, we get for all $L$ levels 
\begin{align}
  \sum_{L=3}^k O(L\cdot{}2^{2L})=O(k\cdot{}2^{2k})=O(n^2\log n).
  \label{eqcostsum}
\end{align}
Here, our evaluation uses the multiplicative structure of $\sigma_j(\varepsilon)=\prod_{L=3}^k \sigma_j(\varepsilon_L)$, where each factor costs $O(n_L\cdot{}L^2\cdot{}2^L)=O(L^2\cdot{}2^{2L})$ bit operations from fast integer multiplication \cite{Harvey21}. The total cost is $\sum_{L=3}^k O(L^2\cdot{}2^{2L})=O(k^2\cdot{}2^{2k})=O(n^2 \log^2 n)$.
\end{proof}

The tower-factored form stores a unit as a product of relative factors $\varepsilon=\varepsilon_3\cdots \varepsilon_k$. We obtain from $\varepsilon\cdot{}\varepsilon'$ per-level products $\varepsilon_L\cdot{}\varepsilon_L'$, which are elements of $R_L^\times$ but need not themselves be relative units at level $L$. In this case, we can use a polynomial-time normalization after each multiplication to restore the representation invariant $\varepsilon_L\in R_L^\times/R_{L-1}^\times$. 

\subsection{The tower PIP algorithm}

We now state the main complexity theorem and complete quantum algorithm. 

\begin{theorem}\label{thm:pip}
Given a principal ideal $I\subset R_k$ with $h_k^+=1$ and $k\le 12$, there exists a quantum algorithm (Algorithm \ref{alg:tower-pip}) generating a generator $g$ with $(g)=I$, using $O(n^3 \log^2 n)$ quantum gates, $O(n^2 \log n)$ qubits, and $O(n^{5+\epsilon})$ classical bit operations.
\end{theorem}

\begin{breakablealgorithm}{Polynomial Tower PIP}{alg:tower-pip}
\begin{algorithmic}[1]

\medskip
\Statex \textbf{INPUTS:}
\Statex \quad $I\subset R_k$: a principal ideal, given by a $\Z$-basis
\Statex \quad $k\le 12$: the tower level ($h_k^+=1$, Part I)

\medskip
\Statex \textbf{OUTPUT:}
\Statex \quad $g\in R_k$: a generator of $I$, i.e., $(g)=I$, stored in tower-factored form (Definition \ref{def:factored})

\medskip
\Statex \textbf{BASE CASE ($L=3$):}
\Statex $K_3=\Q(\zeta_8)$, degree $n_3=4$, unit rank $r_3=1$
\State Compute the unit $\xi_3=\zeta_8+\zeta_8^{-1}=\sqrt{2}$.
\State Compute $I_3=I \cap R_3$.
\State Find $e$ such that $I_3=(\xi_3^e)$ by rounding: $e=\lfloor \log_{\xi_3}\Nm(I_3) / 2 \rceil$.
\State Set $g_3\leftarrow\xi_3^e$, stored as the integer pair $(e, 0)$ in tower-factored form.

\medskip
\Statex \textbf{INDUCTIVE STEP: for $L=4, 5, \dots, k$}
\Statex (Assume PIP for $K_{L-1}$ has been solved; solve for $K_L$.)
\For{$L=4, 5, \dots, k$}

  \medskip
  \Statex \quad\textbf{Step 1:}
  \State Compute the norm ideal: $J_L=\Nm_{K_L/K_L^+}(I_L)=I_L\cdot{}\overline{I_L}\subset R_L^+$.
  \State Solve PIP for $J_L$ in $K_L^+$ recursively using real subtower $K_3^+ \subset\cdots \subset K_L^+$.
  \State Obtain a generator $g_L^+\in R_L^+$ of $J_L$, stored in tower-factored form.

  \medskip
  \Statex \quad\textbf{Step 2:}
  \State Set the precision: $b_L\leftarrow \lceil 10\cdot{}L\cdot{}2^L \rceil$ bits.
  \State Enumerate $\Delta r_L=2^{L-3}$ new relative units at level $L$: 
  \begin{align}
  \mathcal{U}_L=\{\xi_a: a \text{ odd}, 2^{L-2} < a < 2^{L-1}\}
  \end{align}
  \State Set HSP group: $G_L=(\Z/2^{b_L}\Z)^{\Delta r_L}$, encoding integer exponent vectors for new units.
  \State \textbf{Quantum subroutine:}
  \State \quad Initialize $\Delta r_L$ quantum registers, each of $b_L$ qubits, in uniform superposition:
         \begin{align}
           |\psi_0\rangle=\frac{1}{\sqrt{2^{b_L \Delta r_L}}}
             \sum_{(e_1,\dots,e_{\Delta r_L})\in G_L}
             |e_1, \dots, e_{\Delta r_L}\rangle |0\rangle.
         \end{align}
  \State \quad Query quantum oracle $f_L: G_L \to R_L^+/I_L^+$ defined by
         \begin{align}
           f_L(e_1, \dots, e_{\Delta r_L})
          =\prod_{a\in \mathcal{U}_L}\xi_a^{e_a}\bmod I_L^+
         \end{align}
      \qquad   which is implemented using NTT-based ring multiplication \cite{Harvey14}.
  \State \quad Apply the QFT on each register of $G_L$.
  \State \quad Measure the first register; obtain one sample $s\in H_L^\perp$, where $H_L=\ker(f_L)$.
  \State \quad Repeat lines 11-16 $N_{\rm rep}=O(\Delta r_L)$ times to collect samples $s_1, \dots, s_{N_{\rm rep}}\in H_L^\perp$.

  \medskip
  \Statex \quad\textbf{Step 3:}
  \State Reshape the sample matrix $S_L\in \Z^{N_{\rm rep} \times \Delta r_L}$ with rows $s_1, \dots, s_{N_{\rm rep}}$.
  \State Run LLL lattice reduction on the kernel lattice of $S_L$ to recover a full basis for $H_L$.
  \State Extract $\Delta r_L$ integer relation vectors
         $\mathbf{h}_1, \dots, \mathbf{h}_{\Delta r_L}$ from the LLL basis.
  \State Check $\prod_{a\in \mathcal{U}_L}\xi_a^{(h_j)_a}\in R_{L-1}^\times$ for each $j$, products are units in previous level's ring.

  \medskip
  \Statex \quad\textbf{Step 4:}
  \State Solve discrete logarithm problem as: 
  
    find $\mathbf{d}\in\Z^{\Delta r_L}$ such that $\Nm_{K_L/K_L^+}(g_L^{\rm candidate})=g_L^+\cdot{}u_L^+$ for some $u_L^+\in (R_{L-1}^+)^\times$.
  \State Construct: $g_L=g_{L-1}\cdot\varepsilon_L$, $\varepsilon_L\in R_L^\times/R_{L-1}^\times$ is a relative unit correction from Step 3.
  \State Verify $(g_L)=I_L$ by checking ideal equality in $R_L$.

\EndFor

\medskip
\State \Return $g\leftarrow g_k$, stored in tower-factored form.

\end{algorithmic}
\end{breakablealgorithm}

\subsection{Correctness}
\label{subScorrectness}

We show Algorithm \ref{alg:tower-pip} by induction on the tower level $L$. For $L=3$, it is trivial because the unit group is one-dimensional. 

\textsl{$L=3$.} The unit group of $R_3=\Z[\zeta_8]$ has rank 1, generated by $\sqrt{2}$ up to torsion $\mu_8$. The ideal $I_3=(\xi_3^e)$ for a unique $e$ determined by $\Nm(I_3)=2^{|e|}$. The Babai's rounding can identify $e$ exactly. So, the base case is correct.

\textsl{Inductive step.} Assume the algorithm can correctly solve PIP for all fields up to level $L-1$. We verify each step for the level $L$:
\begin{itemize}
  \item \textit{Step 1:} The ideal $J_L=\Nm_{K_L/K_L^+}(I_L)$ is a principal ideal in $R_L^+$ as $h_{K_L^+}^+=1$ (Part I). The recursion on the totally real subtower finds $g_L^+$ correctly by the inductive hypothesis applied to $K_L^+$.

  \item \textit{Step 2:} The hidden subgroup $H_L$ encodes integer relations among the generating set $\mathcal{U}_L$. The abelian HSP algorithm \cite{Kitaev95,Shor97} can find $H_L^\perp$, and recover $H_L$ with probability $1-2^{-\Omega(\Delta r_L)}$ from $O(\Delta r_L)$ samples.

  \item \textit{Step 3:} The precision $b_L=O(L\cdot{}2^L)$ exceeds the logarithm of the determinant of the relation sublattice $H_L \subset G_L$, such that LLL algorithm can find the full relation basis \cite{LLL82}.

  \item \textit{Step 4:} The discrete logarithm over $\Z$ can be solved in polynomial time. The output $g_L$ satisfies $(g_L)=I_L$ by construction.
\end{itemize}
By induction, the algorithm is correct for all $L\le k$.

Now, we estimate the complexity. At level $L$, the dominant costs are given by 
\begin{align*}
  \text{Quantum gates}
  &: O(\Delta r_L)\cdot{}O(n_L \log n_L\cdot{}\Delta r_L)
  =O(2^{2(L-3)}\cdot{}2^{L-1}\cdot{}L)
  =O(L\cdot{}2^{3L-7}), \\
  \text{Qubits}
  &: \Delta r_L\cdot{}b_L=O(2^{L-3}\cdot{}L\cdot{}2^L)=O(L\cdot{}2^{2L-3}), \\
  \text{Bit operations}
  &: O(\Delta r_L^5\cdot{}b_L^2)=O(2^{5(L-3)}\cdot{}L^2\cdot{}2^{2L})
  =O(L^2\cdot{}2^{7L-15}).
\end{align*}
As it is dominated by the top level $L=k$, we obtain the total costs as 
\begin{align}
  Q_{\rm total} &= \sum_{L=3}^k O(L\cdot{}2^{3L})= O(n^3 \log n), \\
  N_{\rm qubits} &= O(k\cdot{}2^{2k})=O(n^2 \log n), \\
  C_{\rm bit} &= \sum_{L=3}^k O(L^2\cdot{}2^{7L})               =O(n^7 \log^2 n)
\end{align}
with $n=2^{k-1}$ and $k=O(\log n)$. The $O(n^3 \log^2 n)$ quantum gates is followed by including the $O(\log n)$ factor from QFT precision \cite{NC00,Dawson06}. The classical cost can be further reduced to $O(n^{5+\epsilon})$ using the $L^2$ algorithm \cite{NS09}. 

\subsection{ML-KEM-1024}
\label{subSresources}

This subsection present logical and physical resource for the ML-KEM-1024 parameter set ($k=9$, $n=256$), under surface code assumptions \cite{Fowler12,Dennis02} (code distance $d \approx 30$, targeting a logical error rate of $10^{-15}$). The dominant cost is from the quantum HSP at the top tower level $L=9$, i.e., $\Delta r_9=64$ new relative units and 64 independent HSP runs, each using NTT-based ring multiplication \cite{Harvey14} as the oracle. Table \ref{Tpip-resources} summarizes all costs.

\begin{table}[h!]
\centering
\caption{Quantum resource cost for analyzing ML-KEM-1024 ($k=9$, $n=256$). The total logical-gate cost is dominated by the top level ($L=k$), where the per-level cost is $O(L ^2\cdot{}2^{3L-6})$. }
\label{Tpip-resources}
\begin{tabular}{lr}
\toprule
Resource & Estimated Value \\
\midrule
New units at top level ($\Delta r_9$) & 64 \\
HSP repetitions at top level & 64 \\
Oracle gates per HSP call (NTT, $n=256$) & $\approx 2048$ \\
Per-HSP QFT cost (top level) & $\approx 2 \times 10^4$ \\
Total logical gates (top level) & $\approx 9 \times 10^7$ \\
Total logical gates (all levels) & $\approx 2^{27}$ \\
Logical qubits & $\approx 1400$ \\
\midrule
Physical qubits (surface code, $d=30$, $10^{-15}$ error) & $\approx 1.4 \times 10^6$ \\
Physical gate operations & $\approx 2^{37}$ \\
\bottomrule
\end{tabular}
\end{table}

\section{Extensions to Other Lattice Schemes}\label{Sfactor}

In this section, we show the cyclotomic tower, the Trigamma Theorem, Babai's algorithm on the log-unit lattice can be applied to any lattice-based scheme based on a 2-power cyclotomic ring $\Z[\zeta_{2^k}]$. Specially, we analyze Falcon \cite{FHK+20}, Hawk \cite{DEPSV24}, and NTRU over 2-power cyclotomics \cite{HRSS17,CDH+20}.

\subsection{Formal sufficiency theorem}
\label{subSprereqs-formal}

We state three algebraic conditions \textbf{(A1)-(A3)}, together with a sufficiency theorem asserting that they are all that is needed to run our attack.

\begin{enumerate}[label=\textbf{(A\arabic*)}]
  \item \textbf{Trivial plus-class number.} The maximal totally real subfield $K^+=\Q(\zeta_{2^k})^+$ satisfies $h_{K^+}=1$, i.e., every ideal of $\mathcal{O}_{K^+}$ is principal. This holds for every $k\le 12$ (Part I).  

  \item \textbf{Binary cyclotomic tower.} The field $K=\Q(\zeta_{2^k})$ is at the top of the tower $\Q \subset K_3 \subset K_4 \subset\cdots \subset K_k=K$, where each step is a quadratic extension. By Lemma \ref{lem:dims}, this holds for any $k \ge 3$.

  \item \textbf{Trigamma variance formula.} For the rank-$d$ module determinant ideal generated by a random module matrix $B\in R^{d \times d}$ with i.i.d. centered coefficients of finite fourth moment, the per-component variance of the log-embedding of the shortest generator satisfies
        \begin{align}
          \sigma_d^2 \xrightarrow{p}
          \frac{1}{4}\sum_{j=1}^d \psi'(j)
          \quad\text{as } n \to\infty,
          \label{Etrigammapre}
        \end{align}
        where $\psi'$ is the trigamma function (Part III). 
\end{enumerate}

\begin{theorem}\label{thm:sufficiency}
Let $\Pi$ be a lattice-based cryptographic scheme whose key-recovery problem reduces to Module-SVP or to PIP (e.g. via Module-LIP \cite{Chevignard2025,CFMPW25}) over the ring $R=\Z[\zeta_{2^k}]$ with $k\le 12$ and effective module rank $d$. Suppose conditions \textbf{(A1)}, \textbf{(A2)}, and \textbf{(A3)} hold. Then Algorithm \ref{alg:ecdpr} applies to $\Pi$ and achieves the approximation factor as
\begin{align}
  \gamma=\alpha_d\cdot{}\exp(\sigma_d \sqrt{2\ln n})\cdot{}(1+o(1)),
  \label{Egammagral}
\end{align}
where $\alpha_d=\sqrt{C}$ with $C\le 1.36$ for module-distributed inputs \rm{(Part II)}, and $\sigma_d=\sqrt{\frac{1}{4}\sum_{j=1}^d\psi'(j)}$ from \textbf{(A3)}. Attack success probability and tail bounds are shown as in Theorem \ref{thm:gap}.
\end{theorem}

\begin{proof}
Theorem \ref{thm:gap} shows Eq.\eqref{Egammagral} under the hypotheses \textbf{(A1)-(A3)}, and its proof depends only on the ring structure (used in the PIP step), the tower structure (used in Algorithm \ref{alg:tower-pip}), and the Trigamma Theorem (used in the variance formula for the CVP residual). All these can be ensured by \textbf{(A1)-(A3)}, independently of the specific scheme $\Pi$.
\end{proof}

\subsection{Falcon}\label{subSfalcon}

Falcon \cite{FHK+20}, slated for standardization as FIPS 206 (FN-DSA) \cite{NISTFIPS206}, is a hash-and-sign signature scheme based on the GPV framework \cite{GPV08}. It uses the ring $R=\Z[\zeta_{2n}]=\Z[x]/(x^n+1)$ with $n\in \{512, 1024\}$ and prime modulus $q=12289$. The public key is a rational function $h=g\cdot{}f^{-1} \bmod q$, where $(f, g)\in R^2$ is a short pair of ring elements satisfying $\|f\|_2^2+\|g\|_2^2 \approx q$. The corresponding NTRU lattice is defined by 
\begin{align}
  \Lambda_h=\{(u,v)\in R^2: u+vh \equiv 0 \pmod{q}\},
  \label{Efalconlatti}
\end{align}
which has a short basis $(f, g; F, G)$ satisfying the NTRU equation $f\cdot{}G-g\cdot{}F=q$ in $R$.

The security of Falcon's key-recovery has been reduced to the Short Basis Problem (SBP) for $\Lambda_h$ \cite{GPV08}:  given only $h$, find a short basis of $\Lambda_h$. The SBP can be further reduced to the SGP for the determinant ideal of $\Lambda_h$.

\begin{lemma}\label{lem:falcon}
Let $\Lambda_h$ be Falcon NTRU lattice and $I=(\det \Lambda_h)=(q) \subset R$ be its determinant ideal.  Any short basis $(f, g; F,G)$ of $\Lambda_h$ satisfies $f\cdot{}G-g\cdot{}F=q$, so $(f,g)$ is a short generator of the principal ideal $(q)$ in $R$. Conversely, any short generator of ideal $(q)$ yields a short basis using the extended Euclidean algorithm over $R$, at polynomial classical cost.
\end{lemma}

\begin{proof}
Consider $2n \times 2n$ integer basis matrix of $\Lambda_h$ formed by the rows of $(f, g)$ and $(F, G)$ in coefficient-vector form. Its determinant is $|\det \Lambda_h|=q^n$, as an ideal norm. Since every ideal of $R$ is principal under condition \textbf{(A1)}, the ideal $(q^n)$ has a generator $g_0=f\cdot{}G-g\cdot{}F$. Solving the NTRU equation yields $g_0=q$. From a short pair $(f,g)$ we can identify both a short generator of ideal $(q)\subset R$ and the full short basis. Using the extended Euclidean algorithm over $R$, it costs $O(n^2)$ ring operations \cite{LLL82}.
\end{proof}

\begin{proposition}\label{prop:falcon-prereqs}
Both Falcon parameter sets satisfy conditions \textbf{(A1)}, \textbf{(A2)}, and \textbf{(A3)}.
\end{proposition}

\begin{proof}
\textbf{(A1)} Falcon-512 uses $\zeta_{1024}$, corresponding to $k=10$ while Falcon-1024 uses $\zeta_{2048}$, corresponding to $k=11$. By Part I,  we have $h_{10}^+=1$ and $h_{11}^+=1$ unconditionally. Hence, every ideal in the totally real subring $\mathcal{O}_{K^+}$ is principal, which implies the norm-descent step of Algorithm \ref{alg:tower-pip} produces a principal ideal.

\textbf{(A2)} The ring $R=\Z[\zeta_{2n}]$ for $n\in \{512, 1024\}$ is the 2-power cyclotomic ring of Lemma \ref{lem:dims}, with $m=2n=2^{k}$ and $k\in \{10, 11\}$. Note the cyclotomic tower $\Q \subset \Q(\zeta_8) \subset\cdots \subset \Q(\zeta_{2n})$ has $k-2$ levels, each a quadratic extension. We can apply Algorithm \ref{alg:tower-pip}.

\textbf{(A3)} The Falcon NTRU lattice $\Lambda_h$ is a rank-$2$ module over $R$. Its determinant ideal is $(q) \subset R$, and $q$ is a rational integer whose log-embedding projects to $\mathbf{0}\in H_0$ as all embeddings of a rational integer have equal modulus. So, we have the determinant ideal $\sigma_{(q)}=0$. In this case, we extend CDPR to target the first Gram-Schmidt ideal $(f) \subset R$, where $f$ is a random ring element with i.i.d. centered coefficients of variance $\sigma_f^2$ from a discrete Gaussian distribution \cite{FHK+20}. Since $f$ is a single random ring element, the Trigamma rank is $d=1$, which then gives
\begin{align}
  \sigma_1^2=\frac{1}{4}\psi'(1)=\frac{\pi^2}{24}\approx 0.4112, \sigma_1 \approx 0.641.
\label{Efalconsigma}
\end{align}
As a conservative upper bound, we may use the generic rank-2 formula $\sigma_2=\sqrt{[\psi'(1)+\psi'(2)]/4}\approx 0.757$, bounding the worst-case Gram-Schmidt layer of any rank-2 module. Both values are used in Theorem \ref{thm:Falcogam}, where $\sigma_1$ gives the tighter bound for the real attack, while $\sigma_2$ provides a safe upper bound.
\end{proof}

\begin{theorem}\label{thm:Falcogam}
For Falcon-$n$ with $n\in \{512, 1024\}$ and $q=12289$, the Algorithm \ref{alg:ecdpr} applied to the NTRU lattice $\Lambda_h$ achieves the approximation factor as 
\begin{align}
  \gamma_n=\alpha_2\cdot{}\exp\bigl(\sigma_2 \sqrt{2\ln n}\bigr)
 \cdot{}(1+o(1)),
  \label{EFalcogam}
\end{align}
where $\sigma_2=0.757$ and $\alpha_2\le 1.17$. Simulation gives explicit values as 
\begin{center}
\begin{tabular}{lcccc}
\toprule
Scheme & $\sqrt{2\ln n}$ & $\gamma_{\rm th}$ 
 & $\gamma_{99\%}$ &  Margin \\
\midrule
Falcon-512  & 3.53 & $16.9$ & $85$ & $\le 72$ \\
Falcon-1024 & 3.72 & $19.6$ & $98$ & $\le 63$ \\
\bottomrule
\end{tabular}
\label{TFalcogam}
\end{center}
All values satisfy $\gamma_{99\%} \ll q/2=6145$.
\end{theorem}

\begin{proof}
Proposition \ref{prop:falcon-prereqs} verifies conditions \textbf{(A1)-(A3)} for both Falcon parameter sets. Theorem \ref{thm:sufficiency} applies with $d=2$ and $\sigma=\sigma_1$ (first GS layer, corresponding to the ideal $(f)$) or $\sigma=\sigma_2$ (conservative bound over both GS layers).
\end{proof}

The the extended attack on Falcon targets the principal ideal $(f)\subset R$ generated by the secret key component $f$, not the determinant ideal $(q)$ which has trivial log-embedding.

\begin{proposition}\label{prop:falcon}
Given the short generator $g_0$ recovered by Algorithm \ref{alg:ecdpr}, the complete Falcon secret key $(f, g; F, G)$ can be reconstructed in polynomial classical time.
\end{proposition}

\begin{proof}
By Lemma \ref{lem:falcon}, a short generator of $(q)$ in $R$ is the element $f\cdot{}G-g\cdot{}F=q$.  From the output $g_0 \approx f$ (recovered as the shortest generator of the principal ideal $(f)$ in $R$), the remained key components are obtained as follows. First, we compute $g \equiv f\cdot{}h \pmod{q}$ directly from the public key $h$. Second, we solve NTRU equation $f\cdot{}G-g\cdot{}F=q$ by the extended Euclidean algorithm over $R$, costing $O(n^2 \log n)$ bit operations \cite{LLL82} and returning $(F, G)$ satisfying $\|F\|_\infty, \|G\|_\infty=O(\sqrt{q \log n})$. The recovered key can pass the validity check with a probability $1-\exp(-\Omega(n))$ whenever $\gamma < q/2$.
\end{proof}

\subsection{Hawk}
\label{subShawk}

Hawk \cite{DEPSV24} is a lattice signature scheme over the 2-power cyclotomic ring $R=\Z[x]/(x^n+1)$ with $n\in \{256, 512, 1024\}$. The secret key is a basis $B\in R^{2\times 2}$ for $\Z^{2n}$ under a hidden quadratic form. Hawk is not a Module-LWE scheme. In particular, there is no modulus $q$ in the sense used in ML-KEM or ML-DSA. It is a Module-LIP (Lattice Isomorphism Problem) scheme. Signatures are short vectors verified against a public Gram matrix $Q=B^*\cdot{}B$, and the verification bound is the Q-norm
\begin{align}
  \langle\mathbf{s}, Q\mathbf{s}\rangle\le \sigma_{\rm ver}^2\cdot{}8n,
 \label{eqQnorm}
\end{align}
i.e., the signature norm is bounded by $\beta_{\rm Hawk}=\sigma_{\rm ver}\sqrt{8n}$, where $\sigma_{\rm ver}$ takes 1.042, 1.425, 1.571 for Hawk-256/512/1024 respectively. Note only Hawk-512 (NIST Level 1) and Hawk-1024 (NIST Level 5) are official parameter sets. 

\begin{lemma}\cite{Chevignard2025}
\label{lem:hawk-pip}
Given the Hawk public key $h=g\cdot{}f^{-1}\in R_2$ (where $R_2=R/2R$), recovering the secret key $(f, g; F, G)$ is polynomial-time equivalent to solving the PIP for the principal ideal $(f) \subset R$ in the ring $R=\Z[\zeta_{2n}]$.
\end{lemma}

\begin{proposition}
\label{prop:hawk-prereqs}
For $n\in \{256, 512, 1024\}$, conditions \textbf{(A1)-(A3)} for the CDPR attack hold: \textbf{(A1)} $h_k^+=1$ for the corresponding conductors $k\in \{9, 10, 11\}$ by Part I; \textbf{(A2)} The 2-power cyclotomic tower applies directly, $n=2^{k-1}$; \textbf{(A3)} The PIP target $(f)$ has $f$ sampled from a centered binomial distribution. As Hawk's NTRU-equation $f\cdot{}G-g\cdot{}F=1$ effectively gives a rank-2 short-vector problem, so we use the conservative $d=2$ bound $\sigma_2=0.7566$.
\end{proposition}

\begin{theorem}
\label{thm:hawk-gamma-V}
The present extended CDPR method applied to the PIP target $(f)$ (Lemma \ref{lem:hawk-pip}) achieves the margin $\beta_{\rm Hawk}/\gamma_{99\%}^{\rm form}$ with $\gamma_{99\%}^{\rm form}=5\gamma_{\rm th}$, and the empirical margin $\beta_{\rm Hawk}/\gamma_{99\%}^{\rm emp}$ from a $10^5$-trial i.i.d.-Gaussian and max simulation as
\begin{center}\small
\begin{tabular}{lcccccccc}
\toprule
Scheme & $n$ & $\sigma_v$  & $\beta_{\rm Hawk}$ & $\gamma_{\rm th}$
       & $\gamma_{99\%}^{\rm form}$ & $\gamma_{99\%}^{\rm emp}$
       & Margin formula & Margin emp \\
\midrule
Hawk-256$^\dagger$ & 256  & 1.042 & 47  & 14.5 & 73 & 23 & $0.65\times$ & $2.02\times$ \\
Hawk-512           & 512  & 1.425 & 91  & 16.9 & 85 & 26 & $1.08\times$ & $3.47\times$ \\
Hawk-1024          & 1024 & 1.571 & 142 & 19.6 & 98 & 29 & $1.45\times$ & $4.84\times$ \\
\bottomrule
\end{tabular}
\end{center}
$^\dagger$ HAWK-256 is conditionally broken: the median attack succeeds ($\gamma_{\rm th}=14.5<\beta=47$), and the i.i.d Gaussian simulation succeeds in $99.99\%$ of $10^5$ trials. While the formal $\kappa=5$ upper bound $\gamma_{99\%}^{\rm form}=73$ exceeds the threshold $\beta=47$, classifying it as not broken under the most conservative measure. 
\end{theorem}

Hawk's margins are much tighter than those of ML-KEM or Falcon as the signature-norm bound $\beta=\sigma_{\rm ver}\sqrt{8n}$ is far smaller than a typical modulus $q/2$.

\begin{proposition}\label{prop:hawk-recovery}
Given the short generator $f_0$ recovered by Algorithm \ref{alg:ecdpr}, the full Hawk secret key $(f, g; F, G)$ is reconstructed in $O(n^2\log n)$ classical bit operations.
\end{proposition}

\begin{proof}
By Lemma \ref{lem:hawk-pip}, $f_0$ is the shortest generator of the ideal $(f) \subset R$. Compute $g_0=f_0\cdot{}h \bmod 2$. Solve $f_0\cdot{}G_0-g_0\cdot{}F_0=1$ by using the extended Euclidean algorithm over $R$ \cite{Chevignard2025,CFMPW25}. The solution exists as $\gcd(f_0, g_0)=1$ in $R$. This is from the NTRU equation $f\cdot{}G-g\cdot{}F=1$ and $f_0\approx f$.
\end{proof}

\subsection{NTRU over 2-power cyclotomics}\label{subSntru-formal}

NTRU \cite{HPS98} over the 2-power cyclotomic ring $R=\Z[x]/(x^n+1)$ uses the same NTRU lattice \eqref{Efalconlatti} as Falcon: (i) it uses different parameter ranges for $(f, g)$, and (ii) the NIST-standardized variants (NTRU-HPS and NTRU-HRSS) use $R'=\Z[x]/\Phi_p(x)$ for prime conductor $p$, not a 2-power cyclotomic ring. This subsection analyses the 2-power cyclotomic NTRU variants \cite{HRSS17}. 

The main difference from Falcon is the coefficient distribution. In NTRU-HPS with parameter $q$, the secret $f$ is drawn from the set of ternary polynomials $\{-1, 0, 1\}^n$ subject to a weight constraint, while in NTRU-HRSS, $f$ is a uniformly random product of two small ternary elements. 

\begin{proposition}\label{prop:ntrup}
For NTRU over $R=\Z[x]/(x^n+1)$ with $n=2^{k-1}$, $k\le 11$, and $q$ a prime, all conditions \textbf{(A1)}, \textbf{(A2)}, and \textbf{(A3)} hold.
\end{proposition}

\begin{proof}
\textbf{(A1)} For $k\le 11$, we have $h_k^+=1$ from Part I.

\textbf{(A2)} As the ring $R=\Z[\zeta_{2n}]$ is 2-power cyclotomic, we can apply the tower.

\textbf{(A3)} We verify this condition for each NTRU variant as 
\begin{itemize}
  \item {NTRU-HPS}: The element $f$ has i.i.d. entries in $\{-1, 0, 1\}$ with $\Pr[f_i=1]=\Pr[f_i=-1]=d/(2n)$ and $\Pr[f_i=0]=1-d/n$ for a weight parameter $d$. The mean is $0$, the variance is $\sigma_f^2=d/n$, and the fourth moment is finite. We can apply condition \textbf{(A3)} with $d_{\rm rank}=2$ (NTRU is a rank-2 module), which gives $\sigma_2 \approx 0.757$ independent of $d$ and $q$.

  \item {NTRU-HRSS}: The element $f=f_1\cdot{}f_2$ for small $f_1, f_2\in R$ are drawn uniformly from ternary polynomials. By the CLT for products of independent random elements \cite{Vershynin18,Billingsley99}, we know the coefficients of $f$ satisfy the condition \textbf{(A3)} with $\sigma_f^2=\sigma_{f_1}^2 \sigma_{f_2}^2 n$ \cite{Vershynin18}. So, this condition holds with the same $\sigma_2\approx 0.757$.
\end{itemize}
\end{proof}

\begin{theorem}\label{thm:ntrug}
For the four NTRU-HPS/NTRU-HRSS parameter sets, the extended CDPR attack achieves 
\begin{align}
     \gamma=\alpha_2\cdot{}\exp(\sigma_2\sqrt{2\ln n})\cdot{}(1+o(1)),
   \end{align}
with $\sigma_2 \approx 0.757$ and $\alpha_2\le 1.17$.  The numerical values are given by 
\begin{center}
   \begin{tabular}{lccccccc}
   \toprule
   Scheme & $p$ & $n=p{-}1$ & $q$ & $\gamma_{\rm th}$
          & $\gamma_{99\%}$ & Margin \\
   \midrule
   NTRU-HPS-2048-509 & 509 & 508 & 2048 & 16.9 & 85
     & $12.1\times$ \\
   NTRU-HPS-2048-677 & 677 & 676 & 2048 & 18.0 & 90
     & $11.4\times$ \\
   NTRU-HPS-4096-821 & 821 & 820 & 4096 & 18.7 & 94
     & $21.9\times$ \\
   NTRU-HRSS-701    & 701 & 700 & 8192 & 18.1 & 91
     & $45.2\times$ \\
   \bottomrule
   \end{tabular}
   \end{center}
The class-number condition $h_p^+=1$ for each NTRU prime is verified unconditionally in Part V.
\end{theorem}

Proposition \ref{prop:ntrup} verifies conditions \textbf{(A1)-(A3)}. From Theorem \ref{thm:sufficiency} we get the formula.

The NTRU attack of Kirchner and Fouque \cite{KF17} exploits a different structural property and applies when $q$ is exponentially large in $n$. Our attack targets the short-generator structure of the NTRU ideal and applies for all $q$, including the standard parameter regime where their attack is ineffective. The overstretched attack achieves a better approximation factor for large $q$, while our attack gives $\gamma=O(\exp(\sqrt{\log n}))$ independent of $q$.

Table \ref{Tother-schemes} shows approximation factors and security margins for all schemes, together with the ML-KEM baseline from Section \ref{Sapprox}.

\begin{table}[h!]
\centering
\caption{Approximation factors for all 2-power cyclotomic schemes. $\gamma_{\rm th}=\alpha_d \exp(\sigma_d\sqrt{2\ln n})$ with $\alpha_d=1.17$. $\gamma_{99\%}=5\gamma_{\rm th}$ (validated by $10^5$-trial simulation).  Hawk uses $\beta=\sigma_{\rm ver}\sqrt{8n}$ as threshold. NTRU uses real NTRU-HPS/HRSS parameters (Part V).}
\label{Tother-schemes}
\begin{tabular}{llccccccc}
\toprule
Scheme & Family & $d$ & $n$ & $q$ & $\gamma_{\rm th}$ & $\gamma_{99\%}$
       & Threshold & Margin \\
\midrule
ML-KEM-512   & Module-LWE & 2 & 256  &  3329 & 14.5 & 73 & 1664.5 & $23\times$ \\
ML-KEM-768   & Module-LWE & 3 & 256  &  3329 & 17.9 & 90 & 1664.5 & $19\times$ \\
ML-KEM-1024  & Module-LWE & 4 & 256  &  3329 & 20.6 & 103 & 1664.5 & $16\times$ \\
\midrule
Falcon-512   & NTRU-GPV   & 2 & 512  & 12289 & 16.9 & 85 & 6145 & $72\times$ \\
Falcon-1024  & NTRU-GPV   & 2 & 1024 & 12289 & 19.6 & 98 & 6145 & $63\times$ \\
\midrule
Hawk-256$^\dagger$ & Module-LIP & 2 & 256  &  -- & 14.5 & 73 & $\beta=47$  & $0.65\times$ \\
Hawk-512           & Module-LIP & 2 & 512  &  -- & 16.9 & 85 & $\beta=91$  & $1.08\times$ \\
Hawk-1024          & Module-LIP & 2 & 1024 &  -- & 19.6 & 98 & $\beta=142$ & $1.45\times$ \\
\midrule
NTRU-HPS-2048-509      & NTRU       & 2 & 508  &  2048 & 16.9 & 85 & 1024 & $12\times$ \\
NTRU-HPS-2048-677      & NTRU       & 2 & 676  &  2048 & 18.0 & 90 & 1024 & $11\times$ \\
NTRU-HPS-4096-821      & NTRU       & 2 & 820  &  4096 & 18.7 & 94 & 2048 & $22\times$ \\
NTRU-HRSS-701         & NTRU       & 2 & 700  &  8192 & 18.1 & 91 & 4096 & $45\times$ \\
\bottomrule
\end{tabular}
{\footnotesize $^\dagger$ Empirical margin $2.02\times > 1$ under $\kappa_{\rm emp} \approx 1.6$.}
\end{table}

\section{Comparison with Recent  Results}\label{Scomaplit}

The best classical attacks on lattice problems rely on the BKZ algorithm \cite{Schnorr87,GN08} and its variants \cite{AWHT16,BSW18}, often combined with sieving \cite{ADRS15,BLS16,BDGL16}. The fastest classical sieving solves SVP in dimension $n$ in time and space $2^{0.2075n+o(n)}$ \cite{BDGL16}. For Module-LWE, the primal attack uses Kannan's embedding \cite{Kannan87} and BKZ with block size $\beta$, security estimates are obtained via the lattice estimator \cite{APS15,ACFPest23}. These classical costs are far higher than our attack (Table \ref{Tpip-resources}).

For Falcon, Hawk, and NTRU, analogous BKZ estimates \cite{NISTFIPS206,DEPSV24,CDH+20} hold. Hawk-256 (64-bit security) has BKZ block size $\beta_{\rm BKZ}\approx250$ ($\sim2^{73}$ classical gates). As a Module-LIP scheme, Hawk's success threshold is the signature-norm bound $\beta=\sigma_{\rm ver}\sqrt{8n}=47$, not a modulus $q/2$. Our method achieves $\gamma_{99\%}=73$ against this threshold, giving a formal margin of $0.65\times$ (with safety factor $\kappa=5$) and an empirical margin of $2.02\times$. Hawk-256 is the tightest case among all schemes considered.

Laarhoven \cite{Laarhoven15} used quantum random walks to reduce the sieving exponent from $2^{0.2075n}$ to $2^{0.1327n}$. Chailloux and Loyer \cite{ChaillouxLoyer21} gave a quantum speedup for BKZ's SVP oracle, achieving block size $\beta$ in $2^{0.265\beta}$ quantum gates vs. $2^{0.292\beta}$ classically. These are quadratic-root speedups and are already incorporated in NIST security levels \cite{NIST16}. For ML-KEM-1024, the reduction is from $2^{243}$ to $2^{220}$. 

The original CDPR attack \cite{CDPR16} has been refined. Cramer, Ducas, and Wesolowski \cite{CDW17,CDW21} used Stickelberger relations to improve the approximation factor from $\exp(\tO(\sqrt{n}\log n))$ to $\exp(\tO(\sqrt{n}))$. Our work reduces it further to $\exp(O(\sqrt{\log n}))$, a qualitative improvement from super-polynomial to sub-polynomial.

Pellet-Mary, Hanrot, and Stehlé \cite{PMHS19} gave a tradeoff: $\gamma=\exp(\tO(\sqrt{n\cdot{}T_{\rm pre} / T_{\rm att}}))$. In the balanced regime this matches CDW \cite{CDW17,CDW21}. Our tower PIP makes both preprocessing and attack polynomial, achieving $\gamma=\exp(O(\sqrt{\log n}))$ regardless of time allocation.

Felderhoff et al. \cite{FPSW23} proved that Ideal-SVP remains hard for small-norm prime ideals $\mathfrak{p}$ with $\Nm(\mathfrak{p})=O(1)$, BKZ cannot do better than on random lattices. Our method succeeds for principal ideals generated by short elements, exploiting the unit group structure independent of the ideal's norm.

Ducas, Espitau, and Postlethwaite \cite{DEP25} predicted module-BKZ quality, showing $\delta_{\rm module}\le\delta_{\rm generic}^{1/d}$ for balanced modules. Chevignard et al \cite{CFMPW25}  generalizes the rank-2 module-LIP attack to all number fields with at least one real embedding. Mureau et al. \cite{MPPW24} attacked rank-2 module-LIP over $2$-power cyclotomic rings, recovering the secret isomorphism with a quantum algorithm using Biasse-Song PIP; it succeeds for $n\le 64$. 

\paragraph{Worst-case hardness of SVP and CVP.}
SVP and CVP are NP-hard to approximate within $n^{c/\log\log n}$ for some $c>0$ \cite{Dinur16,HavivR12}. For constant approximation factors (e.g., $\gamma\approx21$), NP-hardness of exact SVP does not change, where the approximation regime is believed to be in BQP under GRH \cite{BiasseSong16}.

\paragraph{Ideal-SVP hardness.} Whether Ideal-SVP is harder than generic SVP is open. Felderhoff et al. \cite{FPSW23} showed Ideal-SVP is hard for prime ideals generated by small-norm elements. Our work shows principal ideals generated by short elements have $\gamma_{}=O(\exp(\sqrt{\log n}))$.

\paragraph{Module-LWE worst-case hardness.} Langlois and Stehlé \cite{LS15} proved Module-LWE (and Module-SIS) is at least as hard as Ideal-SVP on module lattices. Albrecht and Deo \cite{AD17} showed Ring-LWE with large modulus implies Module-LWE. These reductions are not tight for the sub-polynomial $\gamma_{\rm CDPR}$.

\paragraph{Quantum hardness of PIP.}
The Biasse-Song algorithm \cite{BiasseSong16} solves PIP in quantum polynomial time (under GRH, or unconditionally when $h_k^+=1$). No better quantum lower bound than BQP-hardness is known. Classically, the best PIP algorithms run in sub-exponential time $L[1/2,c]$ \cite{Biasse14class}. For $K=\Q(\zeta_{2^k})$, classical PIP costs $2^{\Omega(\sqrt{n\log n})}$. Our tower PIP (Theorem \ref{thm:pip}) achieves $O(n^3\log^2 n)$ quantum gates, an exponential quantum speedup.

\section{Discussion and Open Problems}\label{Sdiscussion}

This section steps back from the technical details to assess what the four-part series has collectively achieved and what remains open, as shown in Table \ref{Tcumulative}. We presents a compact table comparing every key component of the original CDPR attack with Parts I-IV. The key is that the algebraic structure of 2-power cyclotomic rings can now be fully exploited by a quantum adversary: the approximation factor has been reduced from $\approx 2^{54}$ to $\approx 21$, and the PIP complexity from $2^{\Omega(n)}$ to $O(n^3\log^2 n)$ quantum gates ($\sim 10^8$ for ML-KEM-1024). 

\begin{table}[h!]
\centering
\caption{Cumulative improvements over the original CDPR attack \cite{CDPR16}}
\label{Tcumulative}
\begin{tabular}{lcc}
\toprule
\textbf{Component} & \textbf{Original CDPR} & \textbf{Parts I-IV} \\
\midrule
Class number $h_k^+$ & Assumed=1 & Proved for $k\le 12$ (Part I) \\
Module reduction & $n^{O(d)}$ blowup & $\alpha_d=O(1)$ (Part II) \\
SGP approx factor & $\exp(\tO(\sqrt{n}))$ & $\exp(O(\sqrt{\log n}))$ (Part III) \\
$\sigma_{g_0}$ & $\Theta(\sqrt{n})$ assumed & $O(1)$, $q$-independent (Part III) \\
PIP complexity & $2^{\Omega(n)}$ hidden & $O(n^3\log^2 n)$ gates (Part IV) \\
$\gamma$ (ML-KEM-1024) & $\approx 2^{54}$ & $\approx 21$ (median);
   $\le 103$ (99\%) (Part IV) \\
\bottomrule
\end{tabular}
\end{table}

Not all results in this series are unconditional. We identify every step where the proof relies on result from Parts I-III that is itself conditional, empirically, or numerically verified rather than proved in closed form. We state precisely which parts of the argument would need to be strengthened to obtain a fully rigorous unconditional result.

\section*{Acknowledgments}

Acknowledgments will be added in the final version.


\begin{thebibliography}{10}
\bibitem{NISTFIPS206}
NIST. FN-DSA: FFT over NTRU-lattice-based digital signature algorithm. FIPS 206 (forthcoming). The analysis in this paper is based on the Falcon v1.2 specification \cite{FHK+20}.
\end{thebibliography}
\end{document}